\documentclass{llncs}

\usepackage[]{apacite}

\usepackage[english]{babel}
\usepackage{amsmath,amssymb}
\usepackage{mathrsfs}
\usepackage{graphicx}
\usepackage{subcaption}
\usepackage{dsfont}
\usepackage{latexsym}
\usepackage{url}
\usepackage{times}

\usepackage{algorithmicx}
\usepackage[ruled]{algorithm}
\usepackage{algpseudocode}
\usepackage{todonotes}

\newcommand{\img}{\mathtt{img}}

\newcommand{\R}{\mathcal{R}}

\newcommand{\val}{\mathit{eval}}

\newcommand{\G}{\mathsf{G}}

\newcommand{\A}{\mathcal{A}}

\renewcommand{\img}{\mathtt{img}}
\renewcommand{\val}{\mathtt{val}}
\newcommand{\opt}{\mathtt{opt}}

\newcommand{\sv}{\phi}

\newcommand{\tmymath}[1]{\tilde{\mathds{#1}}}

\newcommand{\tuple}[1]{\langle#1\rangle}
\newcommand{\nop}[1]{}

\newcommand{\longv}[1]{}

\newcommand{\products}{\mathit{products}}
\newcommand{\neigh}{\mathit{Neigh}}

\newcommand{\submitted}{\psi}

\newtheorem{fact}[theorem]{Fact}
\begin{document}

\title{Computing the Shapley Value in Allocation Problems: Approximations and Bounds, with an Application to the Italian VQR Research Assessment Program}


\author{Francesco Lupia\inst{1} \and Angelo Mendicelli\inst{1} \and Andrea Ribichini\inst{2} \and Francesco Scarcello\inst{1} \and Marco Schaerf\inst{3}}

\institute{Dept. of Computer Science, Modeling, Electronics and Systems Engineering, University of Calabria, 87036, Rende, Italy\\
 {\tt \{lupia,a.mendicelli,scarcello\}@dimes.unical.it}
 \and Dept. of Physics, Sapienza University, 00185, Rome, Italy\\
 {\tt ribichini@dis.uniroma1.it}
 \and Dept. of Computer, Control, and Management Engineering Antonio Ruberti, Sapienza University, 00185, Rome, Italy\\
 {\tt marco.schaerf@uniroma1.it}}

\maketitle

\begin{abstract}
In allocation problems, a given set of goods are assigned to agents in such a way that the social welfare is maximised, that is, the largest possible global worth is achieved.
 When goods are indivisible, it is possible to use money compensation to perform a fair allocation taking into account the actual contribution of all agents to the social welfare.
 Coalitional games provide a formal mathematical framework to model such problems, in particular the Shapley value is a solution concept widely used for assigning worths to agents in a fair way.
 Unfortunately, computing this value is a $\#${\rm P}-hard problem, so that applying this good theoretical notion is often quite difficult in real-world problems.

We describe useful properties that allow us to greatly simplify the instances of allocation problems, without affecting the Shapley value of any player. Moreover, we propose algorithms for computing lower bounds and upper bounds of the Shapley value, which in some cases provide the exact result and that can be combined with approximation algorithms.

The proposed techniques have been implemented and tested on a real-world application of allocation problems, namely, the
 Italian research assessment program, known as VQR. For the large university considered in the experiments, the problem involves thousands of agents and goods (here, researchers and their research products).
 The algorithms described in the paper are able to compute the Shapley value for most of those agents, and to get a good approximation of the Shapley value for all of them. 
\end{abstract}


\section{Introduction}

\subsection{Coalitional Game Theory}
Coalitional games provide a rich mathematical framework to analyze interactions between intelligent agents.
We consider coalitional games of the form $\G=\tuple{N,v}$,  consisting of a set $N$ of $n$ agents and a characteristic function $v$. The latter maps each
coalition $C\subseteq N$ to the worth that agents in $C$ can obtain by collaborating with each other. In this context, the crucial problem is to find a mechanism to allocate the worth $v(N)$, i.e., the value of the grand-coalition $N$,  in a way that is fair for all players and that additionally satisfies some further important properties such as efficiency: we distribute precisely the available budget $v(N)$
 to players (not more and not less).
 Moreover, for fairness and stability reasons, it is usually required that every group of agents $C$ gets at least
 the worth $v(C)$ that it can guarantee to the game.

Several solution concepts have been considered in the literature as ``fair allocation'' schemes and, among  them, a prominent one is the \emph{Shapley value}~\cite{shapley53}. According to this notion, the worth of any agent $i$ is determined by considering its actual contribution to all the possible coalitions of agents. More precisely, it is considered the so-called {\em marginal contribution} to any coalition $C$, that is, the difference between what can be obtained when $i$ collaborates with the agents in $C$ and what can be obtained without the contribution of $i$.
   More formally, the Shapley value of a player $i \in N$ is defined by the following
 weighted average of all such marginal contributions:

$$\sv_i(\G)=\sum_{C \subseteq N\setminus \{i\}}\frac{|C|!(n-|C|-1)!}{n!}\Big(v(C\cup \{i\})-v(C)\Big).$$

\subsection{Allocation Games}
Among the various classes of coalitional games, we focus in this paper on \emph{allocation games}, which is a setting for analyzing fair division problems where monetary compensations are allowed and utilities are quasi-linear~\cite{Moulin1992}. Allocation games naturally arise in various application domains, ranging from house allocation to room assignment-rent division, to (cooperative) scheduling and task allocation, to protocols for wireless communication networks, and to queuing problems (see, e.g.,~\cite{Greco2014i,militano:twc,Francois2003,Mishra2007,Moulin1992} and the references therein).

 Computing the Shapley value of such games is a difficult problem, indeed it is
 {\rm \#P}-{hard} even if  goods can only have two different possible values~\cite{GLS15}.
 In this paper we focus on large instances of this problem, involving thousands of agents and goods, for which no algorithm described in the literature is able to provide an exact solution.
  There are however some promising recent advances that identify islands of tractability for the allocation problems where at most one good is allocated to each agent:
  it has been recently shown that those instances where
  the treewidth of the agents' interaction-graph is bounded by some constant (i.e., have a low degree of cyclicity) can be solved in polynomial-time~\cite{GLS15}. The result is based on recent advances on counting solutions of conjunctive queries with existential variables~\cite{Greco2014b}.
  Unfortunately, if the structure is quite cyclic this technique cannot be applied to large instances, because its computational complexity has an exponential dependency on the treewidth.

  In some applications, one can be satisfied with approximations of the Shapley value.
  With this respect, things are quite good in principle, since we know
 there exists a fully polynomial-time randomized approximation scheme
  to compute the Shapley value in supermodular games~\cite{Liben-Nowell2012}.
   The algorithm can thus be tuned to obtain the desired maximum expected error, as a percentage of the correct Shapley value.
   However, not very surprisingly, for very large instances one has to consider a huge number of samples, in order to stay below a reasonable expected error. 
   Maleki et al.~\cite{Maleki2014} provide bounds for the estimation error (as an absolute number rather than a percentage of the correct value) if the variance or the range of the samples are known. They also introduce stratified sampling as a method to further reduce the number of required samples.

 \subsection{Contribution}
 In order to attack large instances of allocation problems, we start by proving some useful properties of these problems that allow us to decompose instances into smaller pieces, which can be solved independently.
  Moreover, some of these properties identify cases where the computation of the worth function can be obtained in a very efficient way.

  With these properties, we are able to use the randomized approximation algorithm of Liben-Nowell et al.~\cite{Liben-Nowell2012} even on instances that (when not decomposed) are very large.

  Furthermore, we note that in some applications one may prefer to determine a guaranteed interval for the Shapley value, rather than one probably good point.
  Therefore, we propose algorithms for computing a lower bound and an upper bound of the Shapley value for allocation problems. In many cases the distance between the two bounds is quite small, and sometimes they even coincide, which means that we actually computed the exact value.
  We also used these algorithms together with the approximation algorithm of Liben-Nowell et al.~\cite{Liben-Nowell2012}, to provide a more accurate evaluation of the maximum error of this randomized solution, for the considered instances.

Moreover, by plugging the computed lower bound values into the randomized sampling algorithm proposed by Maleki et al.~\cite{Maleki2014}, we were able to express their error bound as a percentage of the correct Shapley value, rather than as an absolute number, at least for our test instances. This allowed us to compute approximate Shapley values for our largest test case (namely, the 2011-2014 research assessment exercise of Sapienza University of Rome), within 5\% of the correct value with 99\% probability, in a matter of hours.

 \subsection{The Case Study}
 \label{subsec:case_study}
 We have tested the proposed techniques on large real-world instances of the VQR2011-2014 Italian research assessment exercise. This exercise requires
 every Italian research structure $R$ to select some research products, and
submit them to an evaluation agency called ANVUR. While doing so, the structure $R$ is in competition with all other Italian research structures, as the outcome of the
evaluation will be used to proportionally transfer the funds allocated by the Ministry to support research activities in the next years (until
the subsequent evaluation process). Every structure $R$  is therefore interested in selecting and submitting its best research products. For the sake of simplicity, we next simply speak of publications instead of research products (which can also be patents, books, etc.), and of universities and departments instead of structures and substructures (which can be other research subjects).
The programme is articulated in two phases:
(1)  Based on authors' self-evaluations and on ANVUR guidelines, $R$ selects and submits to ANVUR (at most) two publications for each one of its authors\footnote{There are exceptions to this rule: in specific circumstances, fewer than two publications are expected for some authors. To our ends, this detail is immaterial.}, in  such a way that any product is formally associated with at most one author.
(2)  ANVUR formulates its independent quality judgment about the submitted publications (the score assigned to each publication is currently made known only to its authors), 
and the sum of the scores resulting from ANVUR's evaluation is then the VQR score of $R$. Eventually, $R$ will receive funds in subsequent years proportional to this score.
Furthermore,
ANVUR also published an evaluation of all departments, based on the product scores (the score of each department was computed as the sum of the scores of the products formally assigned to the authors in that department). Finally, the scores were also used for evaluating individual researchers that had been recently hired by $R$ (this also greatly influenced $R$'s funds in subsequent years), as well as those researchers that were members of PhD committees. Scores for recently hired researchers were computed as the sum of the scores of the products formally assigned to them; data in this respect were published by ANVUR in aggregated form only, for each department and for each scientific disciplinary sector. Evaluations for researchers that were members of PhD committees were computed as the sum of the scores of the best publications each one of them had coauthored, among all the publications submitted for the VQR (for this evaluation, the formal assignment of publications to authors was irrelevant); data in this respect were published by ANVUR in aggregated form only, for each PhD committee.

The way ANVUR currently uses product scores, for the purposes described above, yields evaluations that do not satisfy the desirable properties outlined in Section~\ref{sec:properties}.
In order to deal with this issue, we have modeled the problem as an allocation game~\cite{GS13}, with a fair way to divide the total score of the university among researchers, groups, and departments based on the Shapley value. The proposed division rule enjoys many desirable properties,
such as the independence of the specific allocation of research products, the independence of the preliminary (optimal) products selection,
the guarantee of the actual (marginal) contribution, and so on.

\section{Preliminaries}

In the setting considered in this paper, a game is defined by an allocation scenario
$\A=\tuple{N,\mathbb{{G}},\Omega,\val,k}$ comprising a set of
agents $N$ and a set of goods $\mathbb{{G}}$, whose values are given by the function $\val$ mapping each good to a non-negative real number.
 The function $\Omega$ associates each agent with the set of goods he/she is interested in.
Moreover, the natural number $k$ provides the maximum number of goods that can be assigned to each agent. 
Each good is indivisible and can be assigned at most to one player.

For a coalition $C\subseteq N$, a (feasible) allocation $\pi_\A[C]$ is a mapping from $C$ to sets of goods from $\mathbb{{G}}$ such that: each agent $i\in C$ gets a set of goods $\pi_\A(i)\subseteq \Omega(i)$ with $|\pi_\A(i)|\leq k$,
and $\pi_\A(i)\cap\pi_\A(j) = \emptyset$, for any other agent $j\in C$
(each good can be assigned to one agent at most).

We denote by $\img(\pi_\A[C])$ the set of all goods in the image of $\pi_\A[C]$, that is,
$\img(\pi_\A[C])=\bigcup_{i\in C} \pi_\A[C](i)$.
    With a slight abuse of notation, we denote by $\val(S)$ the sum of all the values of a set of goods $S\subseteq \mathbb{{G}}$,
     and by $\val(\pi_\A[C])$ the value $\val(\img(\pi_\A[C]))$.
 An allocation $\pi_\A[C]$ is optimal if there exists no allocation $\pi'_\A[C]$ with
  $\val(\pi'_\A[C]) > \val(\pi_\A[C])$. The total value of such an optimal allocation for the coalition $C$ is denoted by $\opt_\A(C)$. The budget available for $\A$, also called the (maximum) social welfare, is $\opt_\A(N)$, that is, the value of any optimal allocation for the whole set of agents $N$
  (the grand-coalition).
  The coalitional game defined by the scenario $\A$ is the pair $\tuple{N,\opt_\A}$, that is, the game where the worth of any coalition is given by the value of any of its optimal allocations.
Note that $\opt_\A(C)\geq 0$ holds, for each $C\subseteq N$, since the allocation where no agent receives any goods is a feasible one (the value of an empty set of goods is $0$).
The definition trivializes for $C=\emptyset$, with $\opt_\A(\emptyset)=0$.


%

\begin{figure}[h]
\centering
   \includegraphics[width=0.7\textwidth]{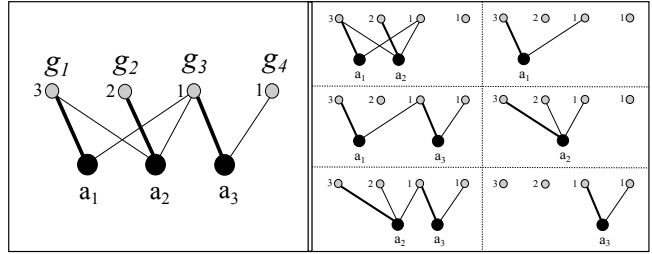}
\caption{Allocation scenario $\A_0$ in Example~\ref{ex:running}.}\label{fig:running}
\end{figure}
\begin{example}\label{ex:running}\em
Consider the allocation scenario $\A_0=\tuple{\{a_1,a_2,a_3\},\{g_1,g_2,g_3,g_4\},\Omega,\val,1}$, depicted in a graphical way in Figure~\ref{fig:running},
where each edge connects an agent to a good she is interested in, and 
 it is possible to allocate just one good to each agent ($k=1$). 
 The figure shows on the left an allocation for all the agents, with the edges in bold identifying
 the allocation of goods to agents.  Note that this is an optimal allocation, i.e., a feasible allocation whose sum of values of the allocated goods is the maximum possible one. The value of this allocation is $\val(g_1)+\val(g_2)+\val(g_3)=3+2+1=6$. 
 
 The coalitional game associated with this scenario is $\G_{\A_0}=\tuple{\{a_1,a_2,a_3\},v_{\A_0}}$,
 where the worth function $v_{\A_0}$ is precisely $\opt_{\A_0}$.
 In particular, we have seen that, for the grand-coalition, $v_{\A_0}(\{a_1,a_2,a_3\})= 6$ holds.
For each $C\subset \{a_1,a_2,a_3\}$ with $C\neq \emptyset$, an optimal allocation restricted to the agents in $C$ is also reported in Figure~\ref{fig:running}. It follows that the other values of the worth function are $v_{\A_0}(\{a_1,a_2\})=5$,
$v_{\A_0}(\{a_1,a_3\})$ = $v_{\A_0}(\{a_2,a_3\})=4$, $v_{\A_0}(\{a_1\}) = v_{\A_0}(\{a_2\})=3$, and $v_{\A_0}(\{a_3\})=1$. \hfill $\lhd$
\end{example}

For any allocation scenario $\A=\tuple{N,\mathbb{G},\Omega,\val,k}$, we define the {\em agents graph} as the undirected graph $G(\A)=(N,E)$ such that $\{i,j\}\in E$ if there is a good $g\in \Omega(i)\cap \Omega(j)$.


\section{The VQR Allocation Game}\label{sec:motivation}

Note that the VQR research assessment exercise can be naturally modeled as an allocation scenario
 $\A=\tuple{\R,\mathcal{P},\products,\val,2}$ where $\R$ is the set of researchers affiliated with a certain university $R$, $\mathcal{P}$ is the set of publications selected by $R$ for the assessment exercise, $\products$ maps authors to the set of publications they have written, and
  $\val$ assigns a value to each publication. In the current VQR programme (covering years 2011-2014), the range of $\val$  is $\{0, 0.1, 0.4, 0.7, 1 \}$, with the latter value reserved to the {\em excellent} products.

  In the submission phase, the values are estimated by the universities according to authors' self-evaluations, and to the reference tables published by ANVUR (not available for some research areas). 
   At the end of the program, $R$ will receive an amount of funds proportional to
  $V_R=\val(\mathcal{P})$, that is, to the considered measure of the quality of the research produced by the university $R$.
  The first combinatorial problem, which is easily seen to be a weighted matching problem, is to
  identify the best allocation scenario for the university. That is, to select a set of publications
  $\mathcal{P}$ to be submitted, having the maximum possible total value among all those authored by $\R$ in the considered period.

   The final result may sometimes be different from the preliminary estimate, in particular because of those publications that undergo a peer-review process by experts selected by ANVUR, which clearly introduces a subjective factor in the evaluation. We assume that the values used by $R$ in the preliminary phase do coincide with the final ANVUR evaluation for all products. This is actually immaterial for the purpose of this paper, because we are interested here in the final division, where only the final (ANVUR) evaluation matters. However, we recall for the sake of completeness that, by adopting the fair division rule used in this paper, the best choice for all researchers is to provide their most accurate evaluation, so that $R$ is able to submit any optimal selection of products to ANVUR. In particular, any strategically incorrect self-evaluation by any researcher is useless, in that it cannot lead to any improvement in her/his personal evaluation, while it can lead to a worse evaluation if the best total value for $R$ is missed~\cite{GS13}.

%
\begin{figure}
  \centering
  \includegraphics[width=0.6\textwidth]{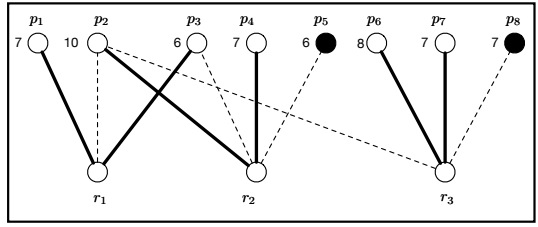}
  \caption{Authors and products in Example~\ref{ex:uno}.}
  \label{fig:esempio}
\end{figure}

\begin{example}\label{ex:uno}\em
Let us consider the weighted bipartite graph in Figure~\ref{fig:esempio}, whose vertices are the researchers $\R= \{r_1,r_2,r_3\}$ of a university $R$ and all the publications they have written.
Edges encode the authorship relation $\products$, and weights encode the 
mapping $\val$ providing the values of the publications.
Consider the optimal allocation $\submitted$ such that $\submitted(r_1)=\{p_1,p_3\}$, $\submitted(r_2)=\{p_2,p_4\}$, and $\submitted(r_3)=\{p_6,p_7\}$, encoded by the solid lines in the figure.
 Based on this allocation, an optimal selection of publications to be submitted for the evaluation is $\mathcal{P}_\submitted=\{p_1,p_2,p_3,p_4,p_6,p_7\}$. The publications that are not submitted are shown in black in the figure.
Note that $p_2$ is co-authored by
$r_1$, $r_2$, and $r_3$, while $p_3$ is co-authored by $r_1$ and $r_2.$
Thus, the allocation scenario to be considered is 
$\A=\tuple{\R,\mathcal{P}_\submitted,\products,\val,2}$, and the associated
coalitional game is the pair $\tuple{\R,\opt(\R)}$.
In particular, the total value of the grand-coalition is $\opt(\R)=45$.
\hfill $\lhd$
\end{example}

 The problem that we face is how to compute, from the total value obtained by $R$, a fair score for individual researchers, or groups, or departments, and so on.  As mentioned above, product scores are currently used for evaluating the hiring policy of universities and the PhD committees, and from this year such scores contribute to evaluate the quality of  courses of study, too. Unfortunately,  this is currently done in a way that fails to satisfy the properties that we outline below. Instead, following~\cite{GS13}, we propose to use the Shapley value of the allocation game defined by the scenario selected by the given structure $R$ as the {\em division rule} to distribute the available total value (or budget) to all the participating agents.
 For the allocation scenario in Example~\ref{ex:uno}, we get
  $\sv_{r_1}=\frac{29}{2}$, $\sv_{r_2}=\frac{29}{2}$, and $\sv_{r_3}=16$. Notice that the Shapley value is not a percentage assignment of publications to authors, but takes into account all possible coalitions of agents.
   Note that $r_3$ is not penalized by the fact that its best publication $p_2$ is assigned to researcher $r_2$, in the submission phase determined by the optimal allocation depicted in Figure~\ref{fig:esempio}.
    Similarly, $r_1$ is not penalized by the fact that the worst publication $p_3$ is assigned to her/him (instead of being assigned to $r_2$).
    
 Another important property is that the value assigned to each researcher is independent by the specific selection of products to be submitted, as long as the submission is an optimal one.
 For instance, an equivalent selection would consist of the products 
 $\mathcal{P}_{\submitted'}=\{p_1,p_2,p_4,p_5,p_6,p_7\}$, because of the optimal allocation $\submitted'$ such that
 $\submitted'(r_1)=\{p_1,p_2\}$, $\submitted'(r_2)=\{p_4,p_5\}$, and $\submitted'(r_3)=\{p_6,p_7\}$.
 It can be checked that no Shapley value changes for any researcher, by considering the alternative allocation scenario $\A'=\tuple{\R,\mathcal{P}_{\submitted'},\products,\val,2}$
 based on the selection of products $\mathcal{P}_{\submitted'}$.
  On the other hand this nice property does not hold for many division rules. For instance, assume that the value  of each researcher is determined by the average score of all the products  evaluated by ANVUR of which she is a (co-)author\footnote{The products that were not submitted 
   cannot be used, because they miss a certified evaluation by ANVUR.}. Then, in the former allocation scenario $r_1$ gets $23/3$, while in the latter one she gets $17/2$.
  Symmetrically, $r_2$ gets a higher value in the former scenario and a lower one in the latter.

  We will now recall the main desirable properties enjoyed by the division rule based on the Shapley value used in this paper. We refer the interested reader to~\cite{GS13} for a more detailed description and discussion of these properties.

 \smallskip {\bf Budget-balance}. \emph{The division rule precisely distributes the VQR score of $ R $ over all its members, i.e.,  $ \sum_{r \in \mathcal R} \sv_r = V_R $.}

\smallskip
{\bf Fairness}. \emph{The division rule is indifferent w.r.t. the specific optimal allocation used to submit the products to ANVUR. In particular,
 the score of each researcher is independent of the particular products assigned to him in the submission phase; moreover, it is independent of the specific set of products $\mathcal{P}$ selected by the university, as long as the choice is optimal (i.e., with the same maximum value $V_R$).}

\smallskip
{\bf Marginality}. \emph{For any group of researchers $ \mathcal S \subseteq \mathcal R $,
$ \sv_{\mathcal S} \ge marg(\mathcal S, \mathcal R)$, where
$ \sv_{\mathcal S} = \sum_{i \in \sv_{\mathcal S}} \sv_i$ and
   $marg(\mathcal S, \mathcal R) = \opt (\mathcal R) - \opt (\mathcal R\setminus \mathcal S)$.
That is, every group is granted at least its marginal contribution to the performance of the grand-coalition $\mathcal R$.}

\smallskip

We remark the importance of the {\em fairness} property, as the choice of a specific optimal set of products is immaterial for $R$, but it may lead to quite different scores for individuals (and for their aggregations, assume e.g. that researchers $r_1$ and $r_2$ above belong to different departments). As a matter of fact, this property does not hold for the division rules adopted by ANVUR for the evaluation of both departments and newly hired researchers (see Section~\ref{subsec:case_study}).  The \emph{budget-balance} property, on the other hand, is violated by the division rule for evaluating researchers who are members of PhD committees.

\section{Useful Properties for Dealing with Large Instances} \label{sec:properties}

Recall that computing the Shapley value is $\#${\rm P}-hard for many classes of games (see, e.g.,
\cite{AzizK14,BachrachR09,Deng1994,Nagamochi1997}), including the allocation games, even if goods may have only two possible values~\cite{Greco2014i}.

For large instances, a brute-force approach is unfeasible, because to compute the value of each agent $i\in N$, it would need to solve $2^n$ optimization problems, where $n=|N|$ is the number of agents.
This is particularly true in our case study, where $n$ is in the order of thousands.


In order to mitigate the complexity of this problem, in this section we will describe some useful properties of the Shapley value, in particular for allocation problems, which allow us to simplify the instances in a preprocessing phase.

Let us consider in this section an allocation scenario $\A=\tuple{N,\mathbb{{G}},\Omega,\val,k}$,
with $\G=\tuple{N,v}$ denoting its associated game, whose agents graph is $G= (N,E)$.
 For such scenario we show the following properties which allow us to simplify the game at hand without altering the Shapley value of any player: \textit{Modularity}, \textit{Null goods}, \textit{Separability}, \textit{Disconnected agent}.


\begin{theorem}[Modularity]\label{fact:modularity}
Let $\{ C_1, C_2 \}$ be a partition of agents of $N$ such that $\Omega(i)\cap \Omega(j)=\emptyset$, for every pair
of agents $i,j$ with $i\in C_1$ and $j\in C_2$. 
 Let $\G_1=\tuple{C_1,v_1}$ (resp., $\G_2=\tuple{C_2,v_2}$) be the coalitional game restricted to agents in  $C_1$
(resp., $C_2$). Then, for each agent $i\in N$, $\sv_i(\G) = \sv_i(\G_1) + \sv_i(\G_2)$.
\end{theorem}

\begin{proof}
Let $ G_1' = \tuple{N, v_1'} $ and $ G_2' = \tuple{N, v_2'} $ be two coalitional games such that,
for each $C\subseteq N$,
$v_1'(C) = v_1 (C \cap C_1)$ and $v_2'(C) = v_2 (C \cap C_2)$. Contrasted with the games in the statement, 
these games are defined over the full set of agents $N$.

Since there are no interactions between agents in $C_1$ and agents in $C_2$, the total value of the optimal allocation for any coalition $C$ is given by the sum of the values of the goods in
the optimal allocations restricted to the two sets of agents  $C \cap C_1$
and $C \cap C_2$. Therefore, we have $ v(C) = v_1' (C) + v_2' (C) $.
Then, from the additivity property of the Shapley value, for each agent $i\in N$, $\sv_i(\G) = \sv_i(\G_1') + \sv_i(\G_2')$. 

Consider now the games $\G_1=\tuple{C_1,v_1}$ and $\G_2=\tuple{C_2,v_2}$) restricted to agents in  $C_1$ and in
 $C_2$, respectively. 
Note that each player $j \in N\setminus C_1$ is \emph{dummy} with respect to the game $G_1'$, so that her Shapley value is null, and her presence have no actual impact on any other player in $G_1'$. 
 In particular such dummy agents could be removed from the game without changing the Shapley value of the other agents, so that for every  $i\in C_1$, we have $sv_i(G_1') = sv_i(G_1) $ and the result immediately follows (by using the same reasoning for $G_2$).

\end{proof}

From the above fact, it follows immediately that each connected component of the agents graph can treated as a separate coalitional game.

\begin{corollary}\label{cor:modularity}
Let  $Z$ be any connected component of the agents graph. The coalitional game
$\G_Z=\tuple{Z,v_Z}$ associated with the allocation scenario obtained by restricting $\A$ to the players in $Z$ 
is such that the Shapley value of each player in $Z$ is the same as in the full game associated with $\A$. 
\end{corollary}

It easy to see that goods having value $0$ do not impact on the computation of the optimal allocation. However, the existence of shared null goods between multiple agents induces connections (among agents) which complicates the structure of the graph.

For instance, consider an allocation scenario $\A'$ comprising three agents $\{r_1,r_2,r_3\}$ having a joint interest only for one good, say $g_0$, whose value is $0$. Any other good has just a single agent interested in it.
  In such a scenario, Corollary~\ref{cor:modularity} cannot be used, since the agents graph associated with the scenario $\A'$  consists of one connected component. On the other hand, without $g_0$, the agents graph would be completely disconnected and thus it would be possible to compute the Shapley values immediately, by using Corollary~\ref{cor:modularity}.
 The following fact states that, in fact, we can get rid of such null goods.

\begin{fact}[No shared null goods]
\label{fact:no-shared-null-goods}
By removing all goods having value $0$ from $\mathbb{{G}}$, we get an allocation scenario with the same associated allocation game. 
\end{fact}
\begin{proof}
Just observe that in the computation of the marginal contribution of any agent $i$ to a coalition $C$, there is no advantage for agents in $C$ in using a good in $\Omega(i)$ having value $0$.
\end{proof}

If it is useful in the algorithms, we can also use Fact~\ref{fact:no-shared-null-goods} in the opposite way, and add 
 null-value goods. Let  $g$ be a good with $\val(g)=0$ and let $X = \{ a \in N \mid g \in \Omega(a)\}$ be the set of agents that are interested in having $g$. Then, the game associated with $\A$ is the same as the game associated with the allocation scenario where $g$ is replaced by fresh goods $g_1,\dots g_{|X|}$ such that each of them is of interest to just one agent in $X$ (hence, there are no connections in the graph because of such goods).


The following property provides us with a powerful simplification method for allocation games. Intuitively, the property states that any set of agents $Z$ that does not exhibit an effective synergy with the rest of the agents can be removed from the game and solved separately.

\begin{theorem}[Separability]\label{thm:separability}
Let $Z$ be any coalition such that $\opt(Z)+\opt(N\setminus Z) \leq \opt(N)$. Then, we can
define from the allocation scenario $\A$  two disjoint allocation scenarios restricted to agents $Z$ and $N\setminus Z$, respectively, that can be solved separately.
For each player $i\in N$, we can compute its Shapley value in the game associated with $\A$ by considering only the game associated with the restricted scenario where $i$ occurs.
\end{theorem}

\begin{proof}
Denote $N\setminus Z$ by $\bar Z$, and consider the allocation games $\G_1=\tuple{Z,v_1}$ and $\G_2=\tuple{\bar Z,v_2}$ restricted to agents in  $Z$ and $\bar Z$, respectively.

Preliminary observe that, for each pair of disjoint coalitions 
$C',C''\subseteq N$, $\opt(C')+\opt(C'') \geq \opt(C'\cup C'')$ holds. Indeed, given any optimal allocation for
 the agents in $C'\cup C''$, its restriction to $C'$ is a feasible allocation for $C'$, as well as 
 its restriction to $C''$ is a feasible allocation for $C''$.  
In particular, we have $\opt(Z)+\opt(\bar Z) \geq \opt(N)$ that, combined with the hypothesis about the considered coalition $Z$, entails that $\opt(Z)+\opt(\bar Z) = \opt(N)$.
This means that the values of the goods not used in any optimal allocation for $\bar Z$ is equal to the sum of the values of the best goods for the agents in $Z$.

We shall show that, for each optimal allocation $\pi$ for $N$, the set of goods $S\subseteq \Omega(Z)$ allocated by $\pi$ to $Z$ is such that $ \val(S) = \opt(Z) $ and the analogous property holds for $\bar Z$. Therefore, these agents get the best goods they can obtain. To prove this claim, consider the value $ v =\val(S) \leq \opt(Z) $ and the value
 $ \bar v \leq opt(\bar{Z}) $.
We know that $ \opt(Z) + \opt(N \setminus Z) = \opt(N) $ and, by the optimality of $\pi $, it holds $v+\bar v=\opt(N)$ too.

Consider now any coalition $C\subseteq N$, and let $C_a = C\cap Z$ and $C_b = C\cap \bar Z$.
  Let $\pi'$ be an optimal allocation for $C$.
We claim that there is an optimal allocation $\pi_a$ mapping goods from $S$ to $Z$ with $\val_{\pi_a}(C_a) = \val_{\pi'}(C_a)$, and an optimal allocation $\pi_b$ mapping goods not in $S$ to $\bar Z$ with $\val_{\pi_b}(C_b) = \val_{\pi'}(C_b)$.
Assume by contradiction that this is not the case. Then at least one of those allocations lead to values smaller than those in $\pi'$ (note that $\pi'$ cannot be worse, because the union of the two restricted allocations is a valid candidate mapping for $C$).
Assume $C_a$ gets a smaller total value (the other case is symmetrical), that is, $\val_{\pi_a}(C_a) < \val_{\pi'}(C_a)$.
Then, there exists some agent $i$ and a good $p\notin S$ so that $p\in \pi'(i)$.
By using Theorem~4.4 in~\cite{Greco2014i}, we can show that this would contradict the fact that $\val(S) = \opt(Z)$.
In fact, goods such as $p$ that are shared with agents outside $Z$ and that allows us to get a better value for the agents in $C_a\subseteq Z$, could be used to improve the choice of the available goods $S$ for the full set $Z$.

Now, given that it suffices to use only the goods in $S$ for $Z$ and the remaining goods for $\bar Z$, we can define an equivalent game in which the goods in $S$ are of interest to agents in $Z$ only and the remaining to agents in $\bar Z$ only.
In the new game, $Z$ and $\bar Z$ are in fact sets of agents with no shared connections and the theorem follows immediately from Theorem~\ref{fact:modularity}.
\end{proof}

A very frequent and important case in applications, which falls in the case considered by this latter property, occurs when $C$ is a singleton $\{i\}$, and it happens that the optimal allocation for this coalition is equal to the marginal contribution of $i$ to $N\setminus \{i\}$. By using the property described above, the set $i$ can be removed from the game and solved separately, so that we immediately get $\phi(i)= \opt(\{i\})$.

 The following property identifies some goods that are useless for some agent $i$ and thus can be safely removed from its set of relevant goods $\Omega(i)$. Note that this operation does not affect other agents possibly interested in such goods. 

\begin{fact}[Useless goods]\label{fact:useless}
 Let $i\in N$ be an agent, and let $g\in \Omega(i)$ be a good such that 
 $\val(g) + \max_{g' \in \Omega(i)\setminus \{g\}}\val(g') <   {\it marg}(\{i\},N) $.
Then, the modified allocation scenario where $g$ is removed from $\Omega(i)$ is equivalent to the original one, that is, the two scenarios have the same associated game.
\end{fact}

We conclude this section with a simple property that does not help to simplify the game, but allows us to avoid the computation of unnecessary optimal allocations, during the computation of marginal contributions.

\begin{fact}[Disconnected agent]\label{fact:disconnected}
 Let $i\in N$ be an agent and let $C\subseteq N$ be a component disconnected from $i$, that is, such that 
  $\Omega(i)\cap \Omega(j) =\emptyset$, for each $j\in C$.
Then, $\opt(\{i\}\cup C)= \opt(\{i\})+\opt(C)$ holds and the marginal contribution of $i$ to $C$ is $opt(\{i\})$.
\end{fact}

\section{Lower and Upper Bounds for the Shapley Value} \label{sec:bounds}

In this section we describe the computation of a lower bound and an upper bound for the Shapley value
of any given allocation game $\G_\A=\tuple{N,v_\A}$.
The availability of such bounds can be helpful to provide a more accurate estimation of the approximation error in randomized  algorithms. Moreover, whenever the two bounds coincide for some agent, we clearly get the precise Shapley value for that agent. 
 We shall see that this often occurs in practice, in our case study.

Preliminarily observe that in allocation games we have for free a simple pair of bounds.
Indeed, recall that the anti-monotonicity property holds, so that, for each pair of coalitions $C_1 \subseteq C_2$,  $marg(\{i\},C_2) \leq marg(\{i\},C_1)$.
Then, for each player $i$ and for every coalition $C\subseteq N\setminus \{i\}$,  
 we have $marg(\{i\},N) \leq marg(\{i\},C) \leq marg(\{i\},\emptyset) = \opt(\{i\})$.
 It immediately follows that 
  $$marg(\{i\},N) \leq \sv_i \leq \opt(\{i\}).$$

 To obtain tighter bounds we observe that the neighbors of $i$ in a coalition $C$ are the agents having the higher influence on the marginal contribution  of $i$ to $C$.
 Indeed, they are precisely those agents interested in using the goods of $i$ when he/she does not belong to the coalition. We already observed that, in the extreme case that no neighbors are present,
  $i$ contributes with all her/his best goods.
  The idea is to consider the power-set of $\neigh(i)$ as the only relevant sets of agents.
     
     Let $P'$ be a set of neighbors of $i$,  and $C = N \setminus (\neigh(i) \cup \{i\})$
     For the computation of the lower bound in Algorithm~\ref{alg:bounds}, for such a profile $P'$
       we compute the marginal contribution of $i$ to  $C\cup P'$, but use this same value
    for the marginal contributions of $i$ to every coalition $C'\subseteq N$ such that 
    $C'\cap \neigh(i)= P'$, that is, for every coalition with the same configuration $P'$ of neighbors of $i$.
  Furthermore, we use a suitable factor $y$ to weigh this value in order to simulate that every such a 
  coalition $C'$ gets that same marginal contribution from $i$.

The case of the upper bound is obtained in the dual way, by using instead the most favorable case where 
 we use the marginal contribution of $i$ to $P'$ in place of the marginal contribution of $i$ to
 any coalition $C'\subseteq N$ with  $C'\cap \neigh(i)= P'$.

%

\begin{algorithm}
	\caption{Computing Bounds for the Shapley Value in Allocation Games}
	\par
	\textbf{Input:} An allocation game $\G_\A=\tuple{N,v_\A}$; \\
	\textbf{Output:} A pair of vectors $(LB,UB)$ encoding, respectively, a lower bound and an upper bound of the Shapley value of $\G_\A$;

	\begin{algorithmic}[1]
		\ForAll {$i \in N$}
			\State $P:=Powerset(\neigh(i))$;
			
			\State $C:=N \setminus (\neigh(i) \cup \{i\})$;
			\State $l= |C|$;
			
			\ForAll {$P' \in P$}
 	\State $Z:= \neigh(i)\setminus P'$;
		\State $y= \sum_{k=0}^{l} {\frac{(l - k + |P'|)! \cdot (|Z| + k)!}{|N|!} \cdot \binom{l}{k}}$;
				
				\State $LB_i\mathrel{+}=y\cdot (v_\A(C \cup P' \cup \{i\}) - v_\A(C\cup P')) $;
				\State $UB_i\mathrel{+}=y\cdot (v_\A(P' \cup \{i\}) - v_\A(P')) $;

			\EndFor		
		\EndFor
		
		\State \Return $(LB,UB)$;
		
	\end{algorithmic}
	\label{alg:bounds}
\end{algorithm}

\begin{theorem} \label {thm:approssimazione}
Let $(LB,UB)$ be the output of Algorithm~\ref{alg:bounds}.
For each agent $i\in N$, $LB_i \leq \sv(i) \leq UB_i$ holds, and the computation of such values can be done in time
 $O(2^{|\neigh(i)|}|N|^3)$.
 \end{theorem}

\begin{proof}
Let $i$ be an agent of the game.
The algorithm is based on the computation of any possible combination $P'$ of the neighbors of $i$.
Regarding the computation of the lower bound, for each such profile $P'$, the algorithm considers a coalition $ C \cup P'$  obtained by completing $P'$ with all the other agents in $N\setminus \{i\}$ that are not neighbors of $i$.

The algorithm uses the value of the marginal contribution of $i$ to such coalition, that is, the value $\delta = v_\A(C \cup P' \cup \{i\}) - v_\A(C\cup P') $, in place of the marginal contributions of $i$ to each coalition $ C' \subseteq N$ such that $ C' \cap\neigh(i) = P' $.
Now, because $C'\subseteq C$, by exploiting the anti-monotonicity property of the marginal contributions in allocation games, we get immediately $marg(\{i \}, C) \leq marg(\{i \}, C')$.
Then, the algorithm weighs in a suitable way $\delta$ so that this value is used in place 
of the right marginal contribution (not lower than $\delta$) of $i$ to each coalition $C'$ of the form described above. A simple combinatorial argument shows that this can be achieved by multiplying $\delta$ by the following factor

  \begin{equation}
y=\sum_{k = 0}^{l}\frac {(l - k + | P'|)! \cdot (| Z | + k) !}{| N |!} \cdot \binom{l}{k},
\end{equation}
where $l= | N \setminus (\neigh(i) \cup \{i\}) |$ and $Z= \neigh(i)\setminus P'$.

Regarding the computation of the upper bound of the Shapley value of $i$, 
we proceed in a similar way but using the marginal contribution of $i$ to the profile $P'$ containing only its neighbors, instead of the marginal contributions to the various coalitions $C' \subseteq N$
such that  $C' \cap\neigh(i) = P' $.
Indeed, in this case we have  $P' \subseteq C'$ and therefore $marg(\{i \}, C') \leq marg(\{i \}, P')$.
Again, we need to multiply such value by a factor which takes into account of all possible ways of completing $P'$ to any coalition $C'$ with the same profile of $i$'s neighbors.
It is easy to see that we can again use the factor $y$ described above, by exploiting the fact that 
$\binom{l}{k} = \binom{l}{l-k}$.

Concerning the computational complexity, just observe that, for each element $P'$ of the power set of $\neigh(i)$, we have to solve a constant number of optimal allocation problems. 
Each of these problems requires the computation of an optimal weighted matching, which can be solved in time $O(|N|^3)$.
\end{proof}

\section{Approximating the Shapley Value} \label{sec:approximation}

\subsection{FPRAS for Supermodular and Monotone Coalitional Games} \label{subsec:approximation-fpras}

In order to approximate the Shapley value, one possibility is to use the Fully Polynomial-time Randomized Approximation Scheme (FPRAS) proposed in~\cite{Liben-Nowell2012}: for any $\epsilon > 0$ and $\delta > 0$, it is possible to compute in polynomial-time an $\epsilon-$approximation of the Shapley value with probability of failure at most $\delta$.
The technique works for supermodular and monotone coalitional games,
 and it can be shown that our allocation games indeed meet these properties~\cite{Greco2014i}.

The method is based on generating a certain number of permutations (of all agents) and computing the marginal contribution of each agent to the coalition of agents occurring before her (him) in the considered permutation. Then the Shapley value of each player is computed as the average of all such marginal contributions. The above procedure is repeated $O(\log(1/\delta))$ times, in indepedent runs, with the result for each agent consisting of the median of all computed values for her (him). Finally, the obtained values are scaled (i.e., they are all multiplied by a common numerical factor) to ensure that the budget-balance property is not violated.

Clearly enough, the more permutations are considered, the closer to the Shapley value the result will be.
We next report a slightly modified version of the basic procedure of this algorithm, where
 we avoid the computation of some marginal contributions, if we can obtain the result
 by using Fact~\ref{fact:disconnected}.

\begin{algorithm}
\caption{Shapley value approximation in allocation games}
\label{alg:rand_opt_sv}
\par
 \textbf{Input:} An allocation game $\G_\A=\tuple{N,v_\A}$;
  \\
 \textbf{Parameters:} Real numbers $0<\epsilon<1$ and $0<\delta<1$;\\ 
 \textbf{Output:} A vector $\tmymath{\sv}$ that is an $\epsilon$-approximation of the Shapley value of $\G_\A$, with probability $1 - \delta$;

\begin{algorithmic}[1]
\State $m = \frac{|N|\cdot(|N|-1)}{\delta \cdot \epsilon^2}$;
\State $i=0$;
\While{$i < m$} \label{alg1:while}
	\State ${\it shuffle}(N)$;
	\State $C:=\{\emptyset\}$;
	\ForAll {$j \in N$}
		\If{$\neigh(j) \cap C \neq \emptyset$}
		\State  $\tmymath{\sv}_j \mathrel{+}= v_\A(C \cup \{j\}) - v_\A(C)$;
		\Else
			\State  $\tmymath{\sv}_j \mathrel{+}= v_\A(\{j\})$;
		\EndIf
		\State $C:= C \cup \{j\}$;
		\State $i = i + 1$;
	\EndFor
\EndWhile
\ForAll {$j \in N$}
	\State  $\tmymath{\sv}_j= \frac{\tmymath{\sv}_j}{m}$;
	\EndFor
\State \Return $\tmymath{\sv}$;

\end{algorithmic}
\end{algorithm}

As a preliminary step, we compute the required number of permutations $m$ to meet the required error guarantee.
In each of the $m$ iterations, the algorithm generates a random permutation from the set of agents $N$.  We then iterate through this permutation and compute the marginal contribution of each agent $j$ to the set of agents $C$ occurring before $j$ in the permutation at hand.
If some neighbor of $j$ (in the agents graph) occurs in $C$, the algorithm proceeds as usual by 
computing the value of an optimal allocation for $C \cup \{j\}$ in order to obtain the value $v_\A(C \cup \{j\})$.
  Note indeed that this one computation is sufficient to get such a marginal contribution,
   because the value $\opt(C)$ for the coalition $C$ including the preceding agents (for the permutation at hand) is known from the previous step.
Moreover, by Fact~\ref{fact:disconnected}, we know that for those permutations in which all the players in $\neigh(j)$ follow $j$, the marginal contribution of $j$ is just $\opt(\{j\})$ (see step 10).
Finally, at steps 16--18 for each agent the algorithm divides the sum of her contributions by the number of performed iterations $m$. The correctness of the whole algorithm follows from Theorem~4 in~\cite{Liben-Nowell2012}.

\noindent \emph{Computation Time Analysis.} Let $n = |N|$ be the number of agents, and let $m$  be the required number of iterations. The cost of the algorithm is $O(m \times n \times margBlock)$, where $margBlock$ denotes the cost of computing each marginal contribution (steps 7--11).
 This requires the computation of an optimal weighted matching in a bipartite graph,
 which is feasible in $O(n^3)$, via the classical Hungarian algorithm.
 However, if the current agent is disconnected from the rest of the coalition,
 the cost is  given by a simple lookup in the cache where the best allocation for each single agent is stored.

\subsection{Sampling Algorithm When the Range of Marginal Contributions Is Known} \label{subsec:approximation-varsamples}

Maleki et al.~\cite{Maleki2014} propose a bound on the number of samples (over the population of marginal contributions) required to estimate an agent's Shapley value, when the range of his/her contributions is known. Their bound is based on Hoeffding's inequality~\cite{Hoeffding1963}, and it states that, in order to approximate the Shapley value of agent $i$ within an absolute value $\epsilon$, with failure probability at most $\delta_{i}$, that is, in order to get

\begin{equation}
\label{eq:maleki-prob-bound}
Prob\{|\tmymath{\sv}_{i}-\sv_{i}|\ge\epsilon\}\le\delta_{i}
\end{equation}

\noindent
at least $m_{i}$ samples are required, where:

\begin{equation}
\label{eq:maleki-sample-bound}
m_{i} = \left \lceil{ \frac{\ln{(\frac{2}{\delta_{i}}) \cdot r_{i}^{2}}}{2 \cdot \epsilon^{2}} }\right \rceil
\end{equation}

\noindent
In the above expression, $r_{i}$ denotes the range of $i$'s marginal contributions (i.e., $r_{i} = \opt(\{i\}) - marg(\{i\},N)$), where $N$ is the set of all agents that partecipate in the allocation game).
 This bound allows us to determine the number of required random samples for each agent $i$, once $\epsilon$ and $\delta_{i}$ are fixed. Assuming we want an overall failure probability $\delta$, each agent $i \in N$ could be assigned a failure probability $\delta_{i}=\delta/|N|$. In principle a higher failure probability 
$\delta_{i}$ could be tolerated for agents with larger ranges, at the expense of lower failure probability for agents with smaller ranges. However, our experimental tests performed with this variant, exhibited just a few marginal gains. 

Once the number of required samples for each agent is determined, the approximate Shapley value, with the desired guarantees on the absolute error, can easily be computed by a randomized algorithm 
evaluating the required samples of coalitions for each player (see Section~\ref{sec:parallel-implementations} for a brief description of our parallel implementation).

In order to consider the classical percentage expression for the approximation error,
 we should replace  $\epsilon$ by $\epsilon \cdot \sv_{i}$ in (\ref{eq:maleki-prob-bound}).
 First observe that $\sv_{i}\neq 0$ for all agents $i$ that are considered by the algorithms, because our simplification techniques preliminarily identify and remove from the game those agents having a null Shapley value (these agents must be interested only in goods with a null value).
In fact, the value of $\sv_{i}$ that would appear in (\ref{eq:maleki-sample-bound}) may be replaced by any known (non-null) lower bound $\ell_i \leq \sv_{i}$, at the expense of taking more samples than those strictly necessary. 
On our largest test instance (namely, the researchers of Sapienza University of Rome who participated in the research assessment exercise VQR2011-2014), the technique described in Section~\ref{sec:bounds} yields lower bounds that are greater that $0$ for all agents. It turns out that, in a matter of hours,
we are able to get approximate Shapley values within $5\%$ of the correct values.

It should be noted that the bound presented by Maleki et al., due to the exponential relation it establishes between $m_{i}$ and $\delta_{i}$,  allows us to compute efficiently good approximate Shapley values, at least on our test instances where the range of the marginal contributions is fairly limited. For a comparison, the FPRAS approach described in Section~\ref{subsec:approximation-fpras} would have taken a few years (instead of the few hours required by the approach presented here) to process our largest input instance with the same error guarantee (see Section~\ref{sec:experiments} for details on our experiments).

\section{Implementation Details and Experimental Evaluation}
\label{sec:experiments}

\subsection{Parallel Implementation of Shapley Value Algorithms}
\label{sec:parallel-implementations}

All the algorithms considered in this paper are amenable to parallel implementation. We engineered our parallel implementations as follows.

\vspace{0.2 cm}
\noindent \emph{FPRAS algorithm}~\cite{Liben-Nowell2012}.  Besides the input allocation game, and the two parameters $\delta$ and $\epsilon$, we added a third parameter, the \emph{thread pool size}. During the execution of the algorithm, each thread (there are as many threads as the thread pool size dictates) is responsible for generating a certain number of permutations according to the requested approximation factor and, for each permutation, it computes the marginal contributions of all authors to that permutation, and saves them to a local cache. Whenever a thread has generated its assigned number of permutations, it delivers its local cache of computed scores to a synchronized output acceptor (which increments the overall score of each author accordingly), and then shuts itself down as its work is completed. When all threads have shut down, each entry of the acceptor's output vector is averaged over the total number of permutations, yielding the final approximate Shapley vector for that run. The above procedure is repeated for each independent run. When all runs are done, the component-wise median of all final approximate Shapley vectors is computed, and the resulting vector is scaled (i.e., all entries are multiplied by a number such that the budget-balance property is enforced), yielding the desired approximation with the desired probability.

\vspace{0.2 cm}
\noindent \emph{Algorithm based on the ranges of samples}~\cite{Maleki2014}. As a preliminary step, the number of required samples for each author $i$ is determined by a sequential routine (as this computation is very fast), based on the approximation parameters $\delta$ and $\epsilon$, and on precomputed values for $\opt(\{i\})$, $marg(\{i\},N)$, where $N$ is the set of all authors, and $LB_{i}$. The algorithm also receives two extra parameters, $threadPoolSize$ and $batchSize$. Subsequently, each thread (the total number of threads is determined by $threadPoolSize$) asks a synchronized producer for a job (i.e., a pair $(i, numSamples)$). The synchronized producer either provides a job for the requesting thread, or it returns \emph{null}, if enough jobs have already been distributed to satisfy the approximation requirements. Upon receiving a job, a thread produces $numSamples$ uniformely distributed random subsets of $N \setminus \{i\}$, and for each such subset $S$, computes the marginal contribution of $i$ to $S$. The sum of these contributions is delivered to a synchronized output acceptor, which stores, for each author, the sum of all marginal contributions computed so far by the various threads. Notice that the job provider will always distribute pairs for which $numSamples \leq batchSize$. This is done to ensure, with proper tuning of parameter $batchSize$, load balancing between the threads. Finally, when a thread receives \emph{null} from the synchronized job provider, it simply shuts itself down, as there is no more work to do. When all threads have shut down, the output acceptor will average the sum of all marginal contributions of each author over the number of required samples for that author, yielding the approximate Shapley value.

\vspace{0.2 cm}
\noindent \emph{Exact algorithm.} In our exact algorithm implementation, each thread (the total number of threads is specified by an input parameter) asks a synchronized producer for a subset of authors to work with. The synchronized subset producer either provides an $n$-bit integer number (where $n$ is the number of authors) for the requesting thread, or it returns \emph{null} if all $2^n$ subsets have already been delivered for elaboration. Upon receiving an $n$-bit integer from the subset provider, a thread turns it into a subset of authors (if a bit is set to 1, then the corresponding author is included in the subset), and computes partial scores for all authors in the subset, storing the values obtained in a local cache. When a thread receives \emph{null} from the subset provider, it delivers its local cache of computed scores to a synchronized output acceptor (which increments the overall score of each author accordingly), and then shuts itself down, as it has no more work to do. When all threads have shut down, the output vector will contain the exact Shapley values for all authors.

\subsection{Experimental Results}
\label{subsec:exp_res}

\vspace{0.2 cm}
\noindent \emph{Hardware and software configuration.}
 Experiments have been performed on two dedicated machines. In particular, sequential implementations were run on a machine with an Intel Core i7-3770k 3.5 GHz processor, 12 GB (DDR3 1600 MHz) of RAM, and operating system Linux Debian Jessie. We tested the parallel implementations on a machine equipped with two Intel Xeon E5-4610 v2 @ 2.30GHz with 8 cores and 16 logical processors each, for a total of 32 logical processors, 128 GB of RAM, and operating system Linux Debian Wheezy.
Algorithms were implemented in Java, and the code was executed on the JDK 1.8.0 05-b13, for the Intel Core i7 machine, and on the OpenJDK Runtime Environment (IcedTea 2.6.7) (7u111-2.6.7-1~deb7u1), for the Intel Xeon machine.

\vspace{0.2 cm}
\noindent \emph{Dataset description.}
 We applied the algorithms to the computation of a fair division of the scores for the researchers of Sapienza University of Rome who participated in the research assessment exercise VQR2011-2014. Sapienza contributors to the exercise were 3562 and almost all of them were required to submit 2 publications for review. We computed the scores of each publication by applying, when available, the bibliographic assessment tables provided by ANVUR. 

\vspace{0.2 cm}
\noindent \emph{Preprocessing.} The analysis was carried out by preliminarily simplifying the input using the properties discussed in Section~\ref{sec:properties}, as explained next.

Starting with a setting with 3562 researchers and 5909 publications, first we removed each researcher having no publications for review. After this step a total of 370 authors were removed. Then, by exploiting the simplification described in Fact~\ref{fact:no-shared-null-goods}, we removed 2323 publications.
By using Theorem~\ref{thm:separability}, the graph was subsequently filtered removing each author whose marginal contribution to the grand coalition coincides with the optimal allocation restricted to the author himself. After this step 2427 researchers out of 3562 were removed.
Then we divided the resulting agents graph into connected components obtaining a total number of 156 connected components and we discovered only two components consisting of more than 10 agents. The sizes of these components are 691 and 15.
Eventually, the components were further simplified by using Fact~\ref{fact:useless}.
After the whole preprocessing phase, we obtained a total of 159 connected components with the largest one having 685 nodes. The size of the second largest component is just 15 while all the others remain very small (less than 10 nodes).
In the rest of the section, we shall illustrate results of experimental activity conducted over the various methods. To this end, we fixed the value $ \delta = 0.01$. This value was chosen heuristically, based on a series of tests conducted on various CUN Areas of Sapienza, where CUN Areas are (large) scientific disciplines such as Math and Computer Science (Area 01) or Physics (Area 02).

\vspace{0.2 cm}
\noindent \emph{Tests with components of variable size.} As already pointed out, after the preprocessing step we obtained very small connected components (less than 10 nodes) except for the largest two (685 and 15 nodes, respectively). For all components with less than 10 nodes, the exact algorithm, of which we used a sequential implementation for these tests,  performs very well (a few milliseconds), therefore we omit the analysis here.
In order to test all the other algorithms, besides the two largest components, we randomly extracted samples of (distinct) nodes out of the original graph, to produce different subgraphs with size $n \in \{23, 26, 30, 40\}$.

\begin{figure}[!t]
\centering
\includegraphics[width=.85\linewidth]{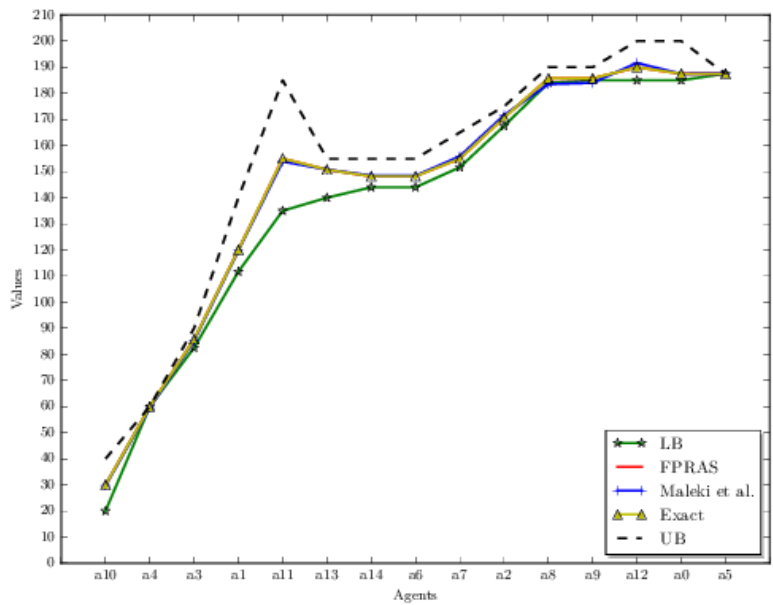}
\caption{Methods comparison ($n=15$).}\label{fig:methodsComparison:roma15sv}
\end{figure}

\begin{figure}[!t]
\centering
\includegraphics[width=.85\linewidth]{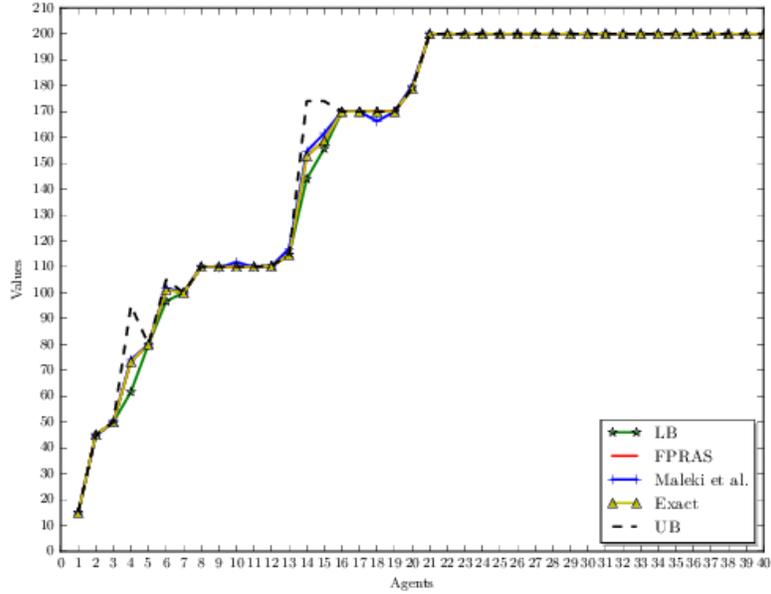}
\caption{Methods comparison ($n=40$).}\label{fig:methodsComparison:roma40sv}
\end{figure}

For the considered cases, 
we do not find significant differences among the values obtained by using the two approximation algorithms and the exact ones   (see, e.g., figures~\ref{fig:methodsComparison:roma15sv} and~\ref{fig:methodsComparison:roma40sv}, in which the approximation algorithms were required to produce results within 5\% of the exact value\footnote{In these two figures
the values obtained by FPRAS are not visible because they coincide with the exact values.}). Notably, with the exception of a small number of cases, our bounds (especially the lower bounds) are always very close to the exact value. In particular, for $n=26$
 we were able to immediately get the Shapley value for all agents, since upper and lower bounds coincide for all of them. 


We also evaluated how many computations of optimal allocations were avoided in the FPRAS of Liben-Nowell et al., by
 exploiting Fact~\ref{fact:disconnected}  (and hence executing in the latter case Step~10 rather than Step~8 in Algorithm~\ref{alg:rand_opt_sv}).  By fixing the approximation error at $\epsilon=0.3$, for each $n \in \{15, 23, 26, 30, 40\}$ we get the following savings:  $9.65 \cdot 10^5$ out of $3.5 \cdot10^6$ (i.e., 28\%), $2.34\cdot 10^6$ out of $1.29\cdot 10^7$ (18\%), $5.36\cdot 10^6$ out of $1.87\cdot 10^7$ (29\%), $8.78\cdot 10^6$ out of $2.9\cdot 10^7$ (30\%), and $1.46\cdot 10^7$ out of $6.93\cdot 10^7$ (21\%), respectively.

As already pointed out, the FPRAS method performed much better than its theoretical guarantee on the maximum approximation error. We report the real maximum and average approximation errors (denoted by X and Y, respectively) of our implementation w.r.t. the exact algorithm for each $n \in \{15, 23, 26\}$, with $\epsilon=0.3$. For $n=15$, we get X = 0.01 and Y = $3\cdot 10^{-3}$, for $n=23$ we get X = $1.5\cdot 10^{-3}$  and Y = $1.7\cdot 10^{-4}$, and for $n=26$ we get X = $1.06\cdot 10^{-4}$ and Y =  $1.59\cdot 10^{-5}$. In all cases, the maximum approximation error was about 1\% (or less) and therefore considerably below the theoretical guarantee (30\%). The algorithm based on the bound of Maleki et al. also performs better than its theoretical guarantee, though not by as wide a margin as the FPRAS method (it is, however, much faster, as we will see in the next paragraph). In this case, for $n=15$ we get X = 0.093 and Y = 0.046, for $n=23$ we get X = 0.098 and Y = 0.011, and for $n=26$ we get X = 0.097 and Y = 0.019. In all cases, the maximum approximation error was below 10\%, and therefore quite smaller than the required threshold.

\vspace{0.2 cm}
\noindent \emph{Running Times.} Figures~\ref{fig:wholeComputationTimes},~\ref{fig:romaParallelTimes} and~\ref{fig:romaMalekiTimes} report the computation times of the various algorithms.
In particular, Figure~\ref{fig:wholeComputationTimes} focuses on the sequential implementations of the brute-force algorithm for computing the exact values, and of the algorithms for computing the upper and lower bounds.
For the experiments, we computed separately the two bounds in order to point out that the computation of the lower bound requires in general more time, because it considers allocation over larger coalitions than those considered for the computation of the upper bound.
Moreover, as discussed in Section~\ref{sec:bounds}, the running times for computing the bounds heavily depend on the cardinality of the agents' neighborhoods. This explains why the running times for the case $n=50$ are smaller than those for the case $n=40$.

\begin{figure}[!t]
\centering
\includegraphics[width=.85\linewidth]{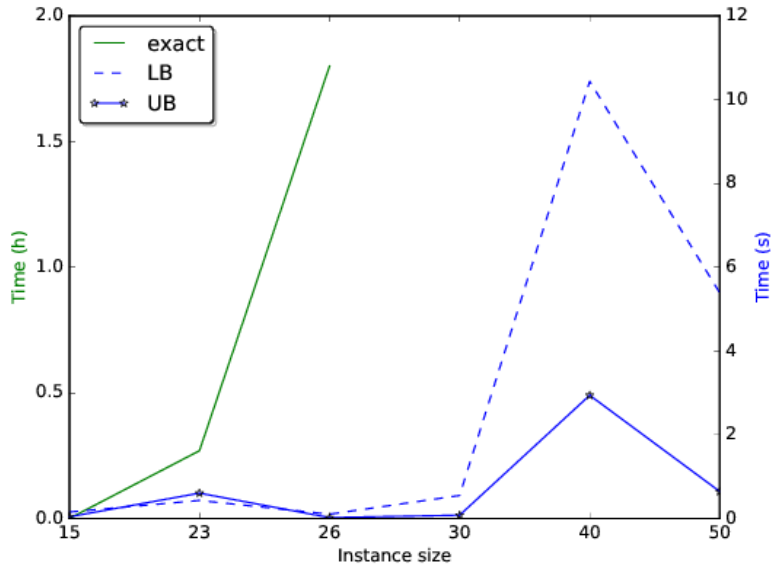}
\caption{Sequential implementations: running times for the computation of the exact value by using the brute-force algorithm (green), and of the upper and lower bounds (blue) vs instance size.}\label{fig:wholeComputationTimes}
\end{figure}

\begin{figure}[!t]
\centering
\includegraphics[width=.85\linewidth]{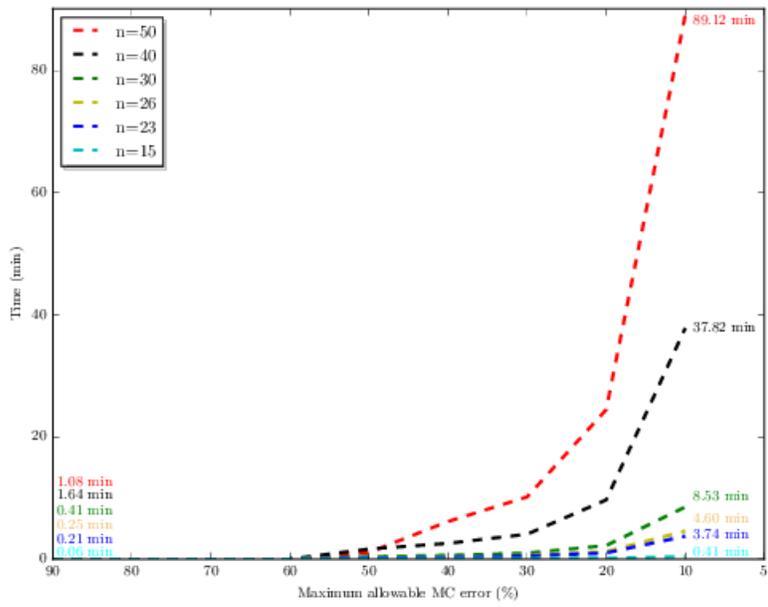}
\caption{Parallel implementation of FPRAS method: running times vs $\epsilon$.}\label{fig:romaParallelTimes}
\end{figure}

\begin{figure}[!t]
\centering
\includegraphics[width=.85\linewidth]{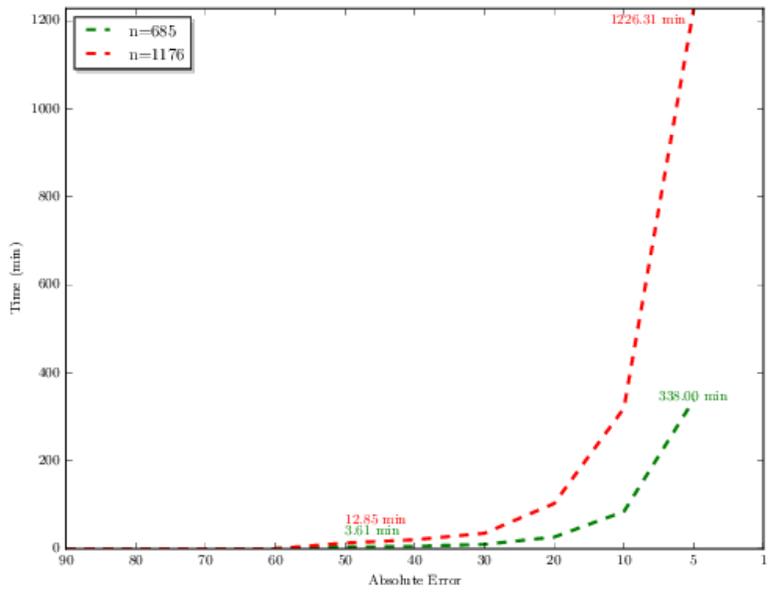}
\caption{Parallel implementation of Maleki-based algorithm: running times vs $\epsilon_{abs}$.}\label{fig:romaMalekiTimes}
\end{figure}

Figure~\ref{fig:romaParallelTimes} shows the running time of the parallel implementation of the FPRAS method, using 24 threads, for different values of $\epsilon$. In particular, we performed five trials over the different (sub)games described above, and report averaged measures. We can see that for games of reasonable size we can achieve a high theoretical approximation error guarantee. For instance, for the largest considered game ($n=50$) we were able to compute the approximate Shapley value with $\epsilon=0.1$ in less than 90 minutes.
There is a big gap between the performances of the FPRAS method, when using the extreme values we considered for the allowed approximation error. However, as  already pointed out, even when we used a poor theoretical guarantee on the approximation error, we still obtained a quite reasonable accuracy.

In spite of its excellent accuracy, and its high efficiency when compared to the exact algorithm, we estimated that our parallel implementation of the FPRAS method would take, with $\epsilon=0.05$ and 24 threads, roughly $3.33$ years to fully analyze the largest component of our Sapienza test case, comprising $685$ authors. By contrast, the parallel implementation of the algorithm based on the bound proposed by Maleki et al., with the same settings, takes only 11.75 hours. The bound on the number of samples proposed by Maleki et al. requires the knowledge of the range of the marginal contributions,
which was computed in less than 3 minutes. Moreover, in order to guarantee that the results are within a certain percentage of the correct values, the lower bounds for the Shapley value are also required. For the biggest component of our test instance, we computed the lower bounds for the 681 authors with neighborhood size up to 19; for the few remaining authors with more neighbors (just 4 authors), we used as lower bound the marginal contribution to the grand coalition.
 Multithreaded computation of the lower bounds took approximately 160 hours.

It should be noted that the bound by Maleki et al. could be applied directly to the largest CC in the unsimplified Sapienza VQR graph. This CC comprises 1176 authors. In this case, straightforward application of the bound for all authors requires, on our server, with 24 threads and an absolute error $\epsilon_{abs}=5$, roughly 20.5 hours. If we set $\epsilon_{abs}=1$, the computation time increases to approximately 31 days. Figure~\ref{fig:romaMalekiTimes} shows the running times of the parallel implementation of Maleki-based algorithm on the two largest CCs in our test instances, with varying values for $\epsilon_{abs}$.

\section{Conclusions and Future Work}\label{sec:conclusion}

In this paper, we have identified useful properties that allow us to decompose large instances 
of allocation problems  into smaller and simpler ones, in order to be able to compute the Shapley value.
 The proposed techniques greatly improve the applicability to real-world problems of the 
  approximation algorithms described in the literature.
 Furthermore, we  described an algorithm for the computation of an upper bound and a lower bound for the Shapley value. 
  These bounds provide a more accurate estimate of approximation errors, and
  (often, in our case study) yield the exact Shapley value for those agents where
  upper and lower bounds coincide.

We have engineered parallel implementations of the considered algorithms, and we have 
tested them on a real-world problem, namely, the 2011-2014 Italian research assessment program (known as VQR), modeled as an allocation game. With the proposed tools, we have been able to compute, either exactly, or within a fairly good approximation (5\% of the correct value with 99\% probability) the Shapley value for all agents in our largest test instance, namely, Sapienza University of Rome, comprising 3562 researchers and 5909 research products.

As future work,
 we would like to extend the structure-based technique described in~\cite{GLS15}
 to the more general class of games where more than one good can be allocated to each agent (as it is the case in VQR allocations). This way, we could compute efficiently the exact Shapley value for large games, provided that the treewidth of the agents graph is small. 
 With this respect, we note that this is not the case for the large Sapienza VQR instance, because after the simplification performed with the tools described in the paper
  we are left with a large component whose estimated treewidth is 64. This is too much for using
  structure-based decomposition techniques. However, for the sake of completeness, we note that all other components have a low treewidth.  For instance, the component with  50 agents used in our tests  has treewidth 5.

Finally, we would like to obtain tighter lower and upper bounds, possibly with a computational effort that can be tuned to meet given time constraints.




\end{document}